\documentclass[11pt]{amsart}
\usepackage{amsfonts,amscd,amsthm,amsgen,amsmath,amssymb, hyperref}
\usepackage[all]{xy}
\usepackage[vcentermath]{youngtab}
\usepackage[letterpaper,hmargin=1.55in,vmargin=1.2in]{geometry}
\usepackage{graphicx, xcolor}
\usepackage{bm}	
\usepackage{amsmath}
\usepackage{amsthm}	
\usepackage{amsfonts}	
\usepackage{amssymb}
\usepackage{mathrsfs}
\usepackage{latexsym}
\usepackage{enumerate}
\usepackage{dsfont}
\usepackage{amssymb}

\theoremstyle{plain}
\newtheorem{theorem}{Theorem}[section]
\newtheorem{corollary}[theorem]{Corollary}

\newtheorem{proposition}[theorem]{Proposition}

\theoremstyle{definition}

\theoremstyle{remark}

\numberwithin{equation}{section}
\numberwithin{figure}{section}

\begin{document}


\title{3D current algebra and twisted K theory}

\author{Jouko Mickelsson}
\address{Department of Mathematics and Statistics, University of Helsinki}

\email{jouko@kth.se}

\maketitle
\begin{abstract} 
Equivariant twisted K theory classes on compact Lie groups $G$ can be realized as
families of Fredholm operators acting in a tensor product of a fermionic Fock
space and a representation space of a central extension of the loop algebra $LG$
using a supersymmetric Wess-Zumino-Witten model. The aim of the present article
is to extend the construction to higher loop algebras using an abelian extension
of a $3D$ current algebra. We have only partial success: Instead of true Fredholm
operators we have formal algebraic expressions in terms of the generators of the
current algebra and an infinite dimensional Clifford algebra. These give rise
to sesquilinear forms in a Hilbert bundle which transform in the expected way
 with respect to $3D$ gauge transformations but do not define true Hilbert
 space operators.

\end{abstract}

\section{Introduction}

I have warm memories from many occasions I had the opportunity to meet Ludwig
Faddeev, mostly in Helsinki and Stockholm. At one point we had very close
mathematical interests although we did not have direct co-operation on the
subject. This contribution to Faddeev's memorial volume is closely related
to that topic. I was working during 1983 on the problem how an additional
Chern-Simons term in the Yang-Mills action affects the current algebra commutation
relations through an abelian extension of the current algebra; the article came
out as preprint in the fall 1983, published later as \cite{Mick2}. I met Ludwig next
spring in Helsinki and I understood that he had been thinking about the same problem
and in fact he had just completed an article with Samson Shatasvili  \cite{F-Sh} 
with similar results. In this contribution the abelian extensions of current algebras which we found at that time  play a central role.

Twisted K theory on compact Lie groups has been discussed from several points of view
in the past. First, let us recall some basic definitions.  Let $PU(H)= U(H)/S^1$ be
the projective unitary group of a complex Hilbert space $H.$ A principal bundle over
a space $X$ with fiber $PU(H)$ is defined up to equivalence by a cohomology class
$\Omega \in \mathrm{H}^3(X,\mathbb{Z}),$ called the Dixmier-Douady class. 
The group $PU(H)$ acts by conjugation on the space of Fredholm operators in $H.$
This action defines a bundle over $X$ with fiber given by the space of Fredholm
operators. The homotopy classes of sections of this fiber bundle define the
twisted K theory $\mathrm{K}^*(X, \Omega).$ This cohomology is $\mathbb{Z}_2$
graded, the even part corresponds to all Fredholm operators on $H$ and the odd part
to self-adjoint Fredholm operators with essential positive and negative spectrum.

In the case when $X=G$ is a compact Lie group one is also interested in the 
equivariant twisted classes, $\mathrm{K}^*_G(G,\mathbb{Z}),$ where the $G$ action on itself
is given by conjugation. The nonequivariant case is solved for compact simple Lie
groups in \cite{Dou}. The equivariant classes can be explicitly constructed using
a supersymmetric conformal field theory model (Wess-Zumino-Witten model) in
$1+1$ space time dimensions \cite{Mick1}. More detailed discussion can be found
in \cite{FHT} including the proof of completeness of the construction.  

The construction in the equivariant case is related to the fact that the compact Lie
group is actually the moduli space of gauge connections on the unit circle $S^1,$
$G =\mathcal{A}/\mathcal{G}_0$ where $\mathcal{A}$ is the space of smooth gauge
connections over $S^1$ and $\mathcal{G}_0$ denotes the group of based gauge
transformations. Thus it is natural to ask if the construction can be generalized
to higher dimensions, to the case of gauge potentials on a vector bundle over
a compact spin manifold. In this article we shall take the modest step and study
the case of the 3-dimensional torus $T^3.$ This case is convenient since we can 
work algebraically with functions which are finite linear combinations of the
Fourier modes. But let us first briefly recall the construction in the case of $X=G$ 
based on loop group representations.

Fix a unitary highest weight representation of the central extension $\widehat{LG}$
of the smooth loop group $LG$ in a Hilbert space $V.$ Use the standard 2-cocycle $k c$ 
with
$$c(X,Y) = \frac{i}{2\pi} \int_{S^1} <X,dY>$$
where  $<\cdot, \cdot>$ is an invariant bilinear form on the Lie algebra $\mathfrak{g}$
such that the level $k$ is a nonnegative integer in an irreducible highest weight
representation of the loop group. In addition one needs an irreducible representation
of the canonical anticommutation relations algebra
$$\psi_a^n \psi_b^m + \psi_b^m \psi_a^n = 2 \delta_{ab} \delta_{n,-m} \bold{1}$$
in a Fock space $\mathcal{F}.$ Here $n,m\in\mathbb{Z}$ and $a,b=1,2,\dots \text{dim\,}
G$ relate to an orthonormal basis of $\mathfrak{g}$ with respect to $<\cdot,\cdot>.$
The representation is characterized by the existence of a vacuum vector with the
property $\psi_a^n |0> =0 $ for $n<0.$ In addition, we have the hermiticity relation
$(\psi^n_a)^* = \psi^{-n}_a.$

Following the notation in \cite{Mick1}  we define the skew symmetric operator
$$Q= \sum \psi_a^n (T_a^{-n} +\frac{1}{3} K_a^{-n})$$
acting in the tensor product $V\otimes \mathcal{F}.$ Here $S^n_a= T^n_a\otimes \bold{1}
+ \bold{1} \otimes K_a^n$ are the generators of the loop algebra acting in the
tensor product. The commutators of the generators $K_a^n$ are determined by
the central extension of level $k=\kappa,$ where $\kappa$ is the dual Coxeter
number of the simple Lie algebra $\mathfrak{g},$

$$[K_a^n, K_b^m] = \sum_c\lambda_{ab}^c K_c^{n+m} +  \kappa m\delta_{ab} \delta_{n,-m}$$
where the numbers $\lambda_{ab}^c$ are the structure constants of the Lie algebra
$\mathfrak{g}.$ 
In a unitary representation of the loop algebra $(T^n_a)^* = -T^{-n}_a$ and
$(K^n_a)^* = -K^{-n}_a.$ 

Next define the family of operators $Q_A = Q + \tilde{k} \sum \psi_a^n A_a^{-n}$ where
the $A_a^n$'s are the Fourier components of the gauge potential $A$ in the Lie
algebra basis labelled by $a$ and $\tilde{k} = k+\kappa.$ One then shows that the family of operators $Q_A$
transforms covariantly under the central extension $\widehat{LG},$
$$ \hat{g}^{-1} Q_A \hat{g} = Q_{A^g}$$
with $A^g = g^{-1} A g +g^{-1}dg$ for $g\in LG.$ Because of the central extension
this family cannot be pushed down to a family of operators parametrized by $\mathcal{A}/
LG$ but it defines a twisted equivariant K theory element on $G.$ The Dixmier-Douady
class of the twist is given by $\tilde{k}$ times the generator 
$\omega\in \mathrm{H}^3(G,\mathbb{Z})\simeq \mathbb{Z}.$

\section{Gauge current algebra in 3 dimension}

Let $M$ be a compact connected spin 3-manifold, $G$ a compact simple Lie
group and $\mathfrak{g}$ its Lie algebra. Let $E$ be a trivial vector bundle over
$M$ with fiber the tensor product $\mathbb{C}^2 \otimes \mathbb{C}^N$ where
the first factor is the space of Weyl spinors and the group $G$ is unitarily
represented in the second factor. The Weyl-Dirac operator $D_A$ coupled to
a gauge potential $A,$ with values in $\mathfrak{g},$ acts in a dense domain
of the Hilbert space $H= L_2(M,E)$ of square integrable sections of the bundle $E.$

For a real number $\lambda$ and a given vector potential $A$ one has a polarization
$$H = H_+(A,\lambda) \oplus H_-(A,\lambda)$$
of the Hilbert space $H$ using the spectral projection associated to $D_A$ with 
a spectral cut at $\lambda.$ This polarization defines a representation of the
canonical anticommutation relations algebra (CAR) in a fermionic Fock space with
a vacuum $|A,\lambda>.$ The canonical anticommutation relations are generated
by elements $a^*(v)$ (linear in $v\in H$) and $a(v)$ (antilinear in $v$) with
relations
$$ a^*(u) a(v) + a(v) a^*(u) = 2 <v,u>\cdot\bold{1} $$
where $<\cdot, \cdot>$ is the $L_2$ inner product in $H.$ The vacuum is a vector
in the Fock space $\mathcal H$ such that
$$ a(v)|A,\lambda> = 0 = a^*(u) |A,\lambda> \text{ for $v\in H_+(A,\lambda)$ and $u\in H_-(A,\lambda)$}.$$

However, there is a problem related to the spectral flow as $A$ varies which
obstructs a construction of the Fock vacua as continuous functions of the argument
$A.$ To circumvent this problem one has to modify the construction by a certain
family of local complex line bundles  over open sets in the space $\mathcal{A}$ 
of smooth vector potentials. The group $\mathcal{G}$ of smooth gauge transformations
does not act in these line bundles; instead, one has to construct an abelian
extension $\hat{\mathcal{G}}$ which acts in the line bundles. The Lie
algebra of the extension is defined by the 2-cocycle
$$c(A; X, Y) = \frac{1}{24\pi^2} \int_M \text{tr}\, A [dX, dY] $$
where $X,Y: M \to \mathfrak{g}$ and the trace is evaluated in the representation in
$\mathbb{C}^N.$ There are several ways to derive this formula but the index theory
derivation in \cite{CMM} describes best the topological and geometric aspects
related to the construction of gerbes over the moduli space of gauge potentials.

As a consequence of the line bundle modifications is that the gauge group $\mathcal G$
acts through an extension $\hat{\mathcal G},$ infinitesimally determined
by the cocycle $c,$ on a vector bundle  $V$ over $\mathcal A$ with model
fiber $\mathcal H.$

Since $c$ is a function of $A$ the 2-cocycle property reads as
\begin{align*}  & c(A; X, [Y,Z]) + c(A; Y, [Z.X]) + c(A; Z,[X,Y]) \\
& +L_X c(A; Y,Z)+L_Y c(A; Z,X) + L_Z c(A, X,Y) =0 \end{align*} 
where $L_X$ denotes the Lie derivative acting on functions of $A$ through
the infinitesimal gauge transformations $L_X A = [A,X] +dX.$ Alternatively, the
cocycle $c$ comes from a central extension of the gauge groupoid $(\mathcal{A},
\mathcal{G})$ with sources and targets in $\mathcal{A}$ and arrows $A\to A^g=
g^{-1}Ag + g^{-1} dg.$ 

\bf Remark \rm  In a real representation of $\mathfrak{g}$ the cocycle $c$ vanishes
identically. The reason is that in a real unitary representation the Lie algebra
elements are antisymmetric and therefore $\text{tr}\, A(BC+CB) = \text{tr}\, (A(BC+CB))^t
= - \text{tr}\, (BC +CB) A = - \text{tr}\, A(BC+CB) =0.$ In particular, this is always
the case for the \bf adjoint representation \rm of $\mathfrak{g}.$

The above remark is important in the next section. We shall consider Fock representations
of the current algebra in the tensor product of two Fock spaces. The first one comes
from the quantization of fermions in the adjoint representation and the second in
an arbitrary complex representation. It follows that the 2-cocycle $c$ has
a contribution only from the complex representation.

\section{Family of formal supersymmetric hamiltonians}

Using the symmetric bilinear form
$$<X,Y> = \frac{1}{(2\pi)^3} \int \text{tr}\, X(z) Y(z) dz$$
on the algebra of smooth functions $X,Y: T^3 \to \mathfrak{g}$ we associate to the Fourier modes $T_a^n$ ($n\in\Bbb Z^3$ and $a=1,2,\dots N= \text{dim}\, \mathfrak{g}$)
the dual basis $\psi_a^n$ and
define the formal series
\begin{equation}
Q= \sum_{n\in \Bbb Z^3}\sum_{a=1}^N\big(\psi^{n}_a T^{-n}_a +  \frac13\psi^{n}_a K^{-n}_a
\big)
\label{freeQ}
\end{equation}
The vectors $\psi_a^{n}$ are interpreted as elements in an infinite dimensional Clifford algebra
with anticommutation relations 
$$\psi_a^{n} \psi_b^m + \psi_b^m \psi_a^n = 2 \delta_{ab} \delta_{n, -m}.$$

Here the $K^{n}_a$'s are  derivations of the CAR algebra defined as
\begin{equation}
[K^{n}_a, \psi^{m}_b] = \sum_c \lambda_{ab}^c \psi^{n+m}_c.\label{Kpsi}
\end{equation}
The commutators of the derivations and are then given as
$$[K^n_a, K^m_b] = \sum_c \lambda_{ab}^c K^{n+m}_c.$$
As infinite formal sums,
$$K^n_a = - \sum_{b,c,m} \frac14 \lambda_{ab}^c  \psi^{m}_b \psi^{n-m}_c.$$
More concretely, in the case of Weyl fermions in the adjoint representation of
the gauge group $G,$ after quantization in the bundle of fermionic Fock spaces $\mathcal{F}$
over $\mathcal{A}$ the current algebra generated by the $K^n_a$'s acts without the abelian extension ($c=0$)
by the remark in the end of the previous Section. For this reason it is actually
convenient to deal with massive 4-component Dirac fermions instead of 2-component
Weyl fermions.  This simplifies some things later on. For massive fermions
we have a mass gap in the spectrum of the Dirac hamiltonians and we can define the vacuum with respect to the polarization defined by the spectral cut at zero for all
Dirac hamiltonians coupled to vector potentials. The fermionic modes $\psi_a^n$
can stand for any of the four components of the Dirac spinor. For notational simplicity we shall suppress the spinor label $\alpha=1,2,3,4$ in the following
discussion. 

The current algebra extension defined by the cocycle $c$ comes entirely from the action
of the $T^n_a$'s on the sections of the bundle $V$ over $\mathcal{A}.$

Since the gauge group acts in the fermionic  Fock spaces $\mathcal F$ without the extension term,
the Dixmier-Douady class of the projective Fock bundle  vanishes and thus can be
pushed to a true Hilbert bundle over the moduli space $\mathcal{A}/\mathcal{G}_0,$
the sections of this bundle corresponding to gauge invariant sections of the
pull-back bundle over  $\mathcal{A}.$ In particular, we are free to define a gauge
covariant vacuum section $A\mapsto |A>$ with $ g |A>=|g\cdot A>.$ For an arbitrary polynomial $P$ in the
generators $\psi^n_a$ the action of a gauge transformation on the state $P|A>$ is then given by the canonical action on the generators of the Clifford algebra (infinitesimally
by the formula \ref{Kpsi}) and by $g|A> ) =|g\cdot  A>.$

As it stands, the infinite sum in \ref{freeQ} is ill-defined as an operator but $Q$ defines 
a derivation of the CAR algebra and the gauge current algebra as
\begin{align}
[Q, \psi^{m}_b ]_+  & = 2(T^m_b +  K^m_b)  \equiv 2 S^m_b  \label{Q,psi}\\
[Q, S^m_b] & = \sum_{n,a} c(A; T^n_a, T^{m}_b) \psi^{-n}_a \label{Q,S}
\end{align}

For a smooth potential $A$ the infinite sum in the second equation is actually
converging in the $L^2$ norm in the CAR algebra generated by the $\psi^n_a$'s
since the Fourier coefficients of a smooth function form a rapidly decreasing 
sequence.

{\bf Remark} An $N\times N$ matrix $X$ defines an element $\psi(X)$ in the dual
$\mathfrak{g}^*$ by $<\psi(X), z> = \text{tr}\, Xz$ where $z$ stands for the matrix representing 
$z\in\mathfrak{g}$ in the representation space $\mathbf{C}^N.$ In the same vain,
a matrix valued 3-form $\omega$ on $T^3$ defines an element in the dual of the
current algebra $Map(T^3, \mathfrak{g})$ by $$<\psi(\omega), X> = \int \text{tr}\,\omega X.$$
It follows that for a linear combination $X= \sum X_m^a S^m_a$ we have
$$[Q, X] = \psi(dX\wedge dA + dA\wedge dX).$$ 

From \eqref{Q,psi} \eqref{Q,S} one can reduce that the commutators with the formal square $Q^2$
are given as
\begin{align}
[Q^2, \psi^n_a] & = 2 \psi^m_b c(A; T^{-m}_b, T^{n}_a) \label{Q2,psi} \\
[Q^2, S^n_a] & = 2 S^m_b c(A; T^{-m}_b, T^n_a) - \psi^m_b \psi^p_c \mathcal{L}^{-p}_c 
c(A; T^{-m}_b, T^n_a). \label{Q2,S}
\end{align}

We have used the fact that
$$\sum_b \mathcal{L}^p_b c(A; T^{-p}_b, T^n_a) = 0$$
which follows from 
$$\mathcal{L}_Z c(A; X,Y) = \int \text{tr}\, [A,Z] (dX dY -dY dX)$$
and the algebraic relation
$$\sum_b \text{tr}\, [A, T_b] (T_b T_a + T_a T_b)
= \sum_b \text{tr}\, [T_a, A]  (T_b)^2 = \sum_b \text{tr}\, T_a[A, (T_b)^2] =0$$
since the element $\sum_b (T_b)^2$ is the Casimir invariant in the semi simple Lie
algebra $\mathfrak{g}.$  

The formal square of $Q$ can be written as 
$$Q^2 = \sum T^n_a T^{-n}_a  +  \frac12 \sum  c(A; T^n_a, T^m_b)  \psi^n_a \psi^m_b$$
and one check by direct computation from the commutation relations that
the commutators of this formal series match the commutators \eqref{Q2,psi} and \eqref{Q2,S}.
This expression can be compared with the corresponding formula \cite{Mick1} in the loop
group case. The important difference is that the contribution involving the dual
Coxeter number is missing. The reason is that the algebra of the $K^n_a$'s does
not contain the anomaly proportional to the dual Coxeter number. 

From  \eqref{Q,psi} follows immediately 
\begin{theorem} \label{gauge}  Let $B = B(A)$ be a matrix value 3-form
on $T^3$ defined as $$B(A) = \omega \wedge dA + dA \wedge\omega + \alpha\wedge A \wedge\beta + \gamma \wedge A
+ A \wedge \phi + \theta$$ where the parameters $\omega, \alpha, \beta$ are matrix valued 1-forms, $\gamma, \phi$ are matrix valued 2-forms and $\theta$ is a matrix valued 3-form.
Thus $B(A)$ is an affine function in the variable $A.$ 
Define the formal  sum  $Q_B = Q + \psi(  B(A)). $ Then for $X= \sum X_m^a S^m_a$
$$[Q_B , X] = \psi(B'(A))$$
where $B'(A)$ is again a matrix valued 3-form with new parameters $(\omega',\alpha',
\beta', \gamma', \phi', \theta').$ In particular,
$$\omega' = [\omega, X] +  dX.$$

 \end{theorem}
Thus $\omega$ transforms like a matrix valued vector potential under the infinitesimal
gauge transformations $X.$ In particular, the action of based gauge transformations
is free on the parameter $\omega$ and thus also on the space of functions
$A\mapsto B(A),$

\begin{corollary} The moduli space $\mathcal{B}/\mathcal{G}_0$ of the functions
$B(A)$ modulo the group $\mathcal{G}_0$ of based gauge transformations is
the classifying space for the group $\mathcal{G}_0$ and thus homotopy equivalent
to the space $\mathcal{A}/\mathcal{G}_0$ of gauge equivalence classes of vector
potentials $A.$ 
\end{corollary}

The forms $Q_B$ are only formal algebraic expressions but they can be used to
define sesquilinear forms in dense domains of the Fock spaces. Let $\xi$ be an arbitrary smooth section of $V$ and $\eta(A) = |A>$ the vacuum section of $\mathcal F.$
The domain $D$ for $Q_B$ is then defined to consists of sections of $V \otimes \mathcal F$
of the form $\xi \otimes p(\psi)\eta$ where $p$ is a polynomial in the generators
$\psi^n_a$ of the Clifford algebra. 
\begin{proposition}
With the notation above  $<\xi'\otimes p'(\psi)\eta,
Q_B (\xi\otimes p(\psi)\eta>$ is well-defined in the fiberwise inner product for the
sections.  
\end{proposition}

\begin{proof}
The vacuum section $\eta$ is invariant under gauge transformations so
$K^n_a \eta =0$ and thus also $\sum \psi^n_a K^{-n}_a \eta =0$ and so
$Q (\xi\otimes \eta) = \sum T_a^n \xi \otimes \psi_a^{-n}\eta.$ Using the algebra
relations \eqref{Q,psi}, \eqref{Q,S} one sees that it is sufficient to show
that the inner products $<\xi'\otimes p'\eta, T_a^n\xi \otimes \psi_a^{-n} p\eta>$
exist fiberwise. First computing this at the fiber over $A=0$ the claim follows
from the fact that for large enough absolute value of $n$ the inner product
of $\psi_a^n$ between finite polynomials in the $\psi$'s vanish. This is a consequence
of the fact that the free vacuum is an eigenvector for the quantized momentum
vectors $P_k$ with eigenvalue zero and thus the vacuum expectation value
of a polynomial not commuting with the $P_k$'s is equal to zero.

When $A\neq 0$ the above argument needs a modification.

To be more precise on the definition of the momentum operators:  These are not the 
free momentum operators but momenta defined with respect to the vacuum $|A>.$ That is,
choosing a unitary operator $T_A$ in the 1-particle Hilbert space such that
$T_A^{-1} \epsilon_0 T_A = \epsilon_A$ where  $\epsilon_A = D_A/|D_A|$ then
$p_k = T_A^{-1} i\frac{\partial}{\partial x_k} T_A$ in the 1-particle space and 
$P_k$ is the second quantized operator corresponding to $p_k.$ 
Here $T_A$ is a pseudodifferential operator with an asymptotic expansion
$T_A = 1 + t_{-1} + t_{-2} + \cdots$ in inverse powers  of momenta; the first 
nontrivial term $t_{-1}$ is linear in $A,$ \cite{JM94}. It follows that
$p_k = i\frac{\partial}{\partial x_k} +$ terms of order $-1$ in momenta.
Now the commutator $[P_k, \psi^n_a]$ is note any more equal to $n_k\psi^n_a$ 
but there is correction which is proportional to the inverse of $|n|.$ It follows
that the expectation value of $\psi^n_a$ between polynomial states in the Clifford
algebra is not exactly zero for large $|n|$ but there is a correction proportional
to $1/|n|^2$ (since the extra terms in $P_k$ have relative magnitude of order $-2$).
But the factor $1/|n|^2$ guarantees that the sum over $n$ in $\sum \psi^n_a T^{-n}_a$
gives a finite fiberwise expectation value between sections of the form $\xi\otimes
p(\psi)\eta.$ The second term $\psi^n_a K^{-n}$ gives a finite contribution
by ${K^n}_a \eta=0$ and the relation 
\eqref{Q,psi}. The interaction term $\psi(B(A))$
is in fact a fiberwise well-defined operator since it is a linear combination
of the $\psi^n_a$'s with $L_2$ coefficients coming from a smooth form $B(A).$ 

\end{proof}

\end{document}